\documentclass[11pt]{article}

\usepackage{amsfonts}
\usepackage{booktabs}
\usepackage{epsfig}
\usepackage{rotating}
\usepackage{amsmath}
\usepackage{geometry}
\usepackage{fullpage}
\usepackage{graphicx,color}
\usepackage{url}
\usepackage{amssymb}

\newtheorem{theorem}{Theorem}

\newtheorem{conjecture}[theorem]{Conjecture}
\newtheorem{corollary}[theorem]{Corollary}

\newtheorem{definition}[theorem]{Definition}

\newtheorem{lemma}[theorem]{Lemma}

\newtheorem{proposition}[theorem]{Proposition}

\newenvironment{proof}[1][Proof]{\noindent\textbf{#1.} }{\ \rule{0.5em}{0.5em}}

\def\01{\{0,1\}}

\newcommand{\ket}[1]{|#1\rangle}
\newcommand{\bra}[1]{\langle#1|}
\newcommand{\ketbra}[2]{|#1\rangle\langle#2|}

\newcommand{\Tr}{\mbox{\rm tr}}

\newcommand{\F}{\mathbb{F}}

\newcommand{\C}{\mathbb{C}}

\newcommand{\al}{\alpha}

\begin{document}


\title{Exclusivity structures and graph representatives of local complementation orbits}


\author{
Ad\'{a}n Cabello
 \thanks{Departamento de F\'{\i}sica Aplicada II, Universidad de Sevilla, E-41012 Sevilla, Spain. \mbox{adan@us.es}}
 \and
Matthew G. Parker
 \thanks{Department of Informatics, University of Bergen, P.O. Box 7803, Bergen N-5020, Norway. \mbox{Matthew.Parker@ii.uib.no}}
 \and
Giannicola Scarpa
\thanks{CWI, Science Park 123, 1098 XG Amsterdam, the Netherlands. \mbox{g.scarpa@cwi.nl}}
\and
Simone Severini
 \thanks{Department of Computer Science, and Department of Physics \& Astronomy, University College London, WC1E 6BT London, United Kingdom. \mbox{simoseve@gmail.com}}
 }

\maketitle


\begin{abstract}
We describe a construction that maps any connected graph $G$ on three or
more vertices into a larger graph, $H(G)$, whose independence number is strictly smaller
than its Lov\'asz number which is equal to its fractional packing
number. The vertices of $H(G)$ represent all possible events
consistent with the stabilizer group of the graph state associated to
$G$, and exclusive events are adjacent. Mathematically, the graph
$H(G)$ corresponds to the orbit of $G$ under local complementation.
Physically, the construction translates into graph-theoretic terms the
connection between a graph state and a Bell inequality maximally
violated by quantum mechanics. In the context of zero-error information theory, the construction suggests a protocol
achieving the maximum rate of entanglement-assisted
capacity, a quantum mechanical analogue of the Shannon capacity, for each $H(G)$.
The violation of the Bell inequality is expressed by the one-shot version of this capacity being strictly larger than the
independence number. Finally, given the correspondence between graphs
and exclusivity structures, we are able to compute the
independence number for certain infinite families of graphs with the use of quantum non-locality, therefore highlighting an application of quantum theory in the proof of a purely combinatorial statement.
\end{abstract}


\section{Introduction}


Partitioning a phase space into orbits is a central step in the study of any
physical or formal dynamics. Immediately after the introduction of graph
(/stabilizer) states in quantum coding theory \cite{CRSS96, gottesman97,sli01} and in the
context of measurement-based quantum computation \cite{raussendorf03}, it
was evident that the related task, when considering the dynamics at the
subsystems level, requires approaches of combinatorial flavour.
While substantial attention has been given to orbits obtained by the
application of local unitaries (with a clear motivation coming from the
classification of multipartite entanglement \cite{hein03}), the question to
decide whether two graph states are equivalent under the action of the local
Clifford group has been settled, by showing \cite{vandennest04} that the equivalence classes are
in one-to-one correspondence with local complementation orbits (also called \emph{Kotzig orbits} \cite{kotzig68}). In
other words, the existence of a sequence of local complementations relating
the associated graphs guarantees equivalence under local Clifford operations
and \emph{viz}.

Even though this link does not embrace full local unitary equivalence, having now counterexamples to the LU-LC conjecture \cite%
{ji10, sarvepalli10}, it unveils a rich interface between the structure of
useful multi-qubit systems and a number of mathematical ideas. Indeed, local
complementation (or, equivalently, $\kappa $-transformation) is a
fundamental operation for studying circle graphs \cite{kotzig68}. This notion has
been instrumental for unifying certain properties of Eulerian tours and
matroids via isotropic systems \cite{bouchet88}, constructs associated to vector spaces over $GF(2)$; and it appears in string
reconstruction problems (related to DNA sequencing) and graph polynomials
\cite{arratia04}.

Given an equivalence class induced by local complementation, in the present
work we shall describe a method for constructing a larger graph associated
to the equivalence class. The method makes use of the stabilizer group of an arbitrary
graph state from the class.
Each of these graphs is identified with an exclusivity structure and a
related \emph{non-contextuality inequality} (for short, NC 
inequality). 
Such an inequality is an upper bound on the sum of probabilities of a set of events, with some exclusivity constraints (a technical discussion about events and exclusivity will be given in Section \ref{con}).
NC inequalities are satisfied by any
non-contextual hidden variable theory, \emph{i.e.}, any physical theory for which the probability of seeing an event is independent of the choice of measurements. Quantum mechanics or more general theories can violate such inequalities. For more details, see \cite{cabello10}.
In this reference, the graph (and more generally an hypergraph), whose
vertices are events, is employed to characterize the correlations for classical and
general probabilistic theories satisfying that the sum of probabilities of pairwise exclusive events cannot be larger than 1. The maximum values for the three physical theories:
classical, general, and quantum, were computed through the three well-known combinatorial parameters: the independence number, the
fractional packing number, and the Lov\'{a}sz number, respectively. As a consequence, it becomes evident that quantum and general probabilistic correlations satisfying that the sum of probabilities of pairwise exclusive events cannot be larger than 1 have semidefinite and linear characterizations, respectively. Quantum mechanics is \emph{sandwiched} between the other two theories.

The framework introduced in \cite{cabello10} permits to quantitatively discuss classical, quantum, and
more general theories through the analysis of a single mathematical object and to have a general technique to single out quantum correlations with \emph{ad hoc} degree of contextuality. For example, a generic graph with independence number strictly smaller than the Lov\'{a}sz
number is associated to a NC inequality violated by quantum
mechanics. If, in addition, the graph has equal Lov\'{a}sz and fractional packing number, then it can be
associated to a NC inequality that is \textquotedblleft
maximally violated\textquotedblright\ by quantum mechanics, meaning that no
general probabilistic theory satisfying that the sum of probabilities of pairwise exclusive events cannot be larger than 1 can achieve a larger value. (Of course, there are graphs for which all theories coincide, as, for example, perfect graphs.)

In the present paper, we propose a construction that translates into the combinatorial language developed in \cite{cabello10} the connection between every graph state of three or more qubits and a Bell inequality maximally violated by quantum mechanics found in \cite{gunhe04}. Namely, we describe a construction that maps any graph on three or more vertices $G$ into a larger graph, $H(G)$, such that its independence number is strictly smaller than its Lov\'asz number which is equal to its fractional packing number. The vertices of $H(G)$ represent all possible events consistent with the stabilizer group of the graph state associated to $G$ and exclusive events are adjacent.

The construction has also applications in zero-error information theory. It leads to a straightforward protocol achieving the maximum rate of zero-error entanglement-assisted capacity \cite{CLMW10, DSW10}. We conjecture that this quantity for a graph $H(G)$ is always strictly larger than its Shannon capacity. A proof of this statement would possibly require a rank bound \emph{a la Haemers}. While it is difficult to compute this bound in general, it may be easier in our case, since $H(G)$ has a very particular structure because of the connection with the stabilizer group. The violation of the Bell inequality is here expressed by the one-shot version of this capacity being strictly larger than the independence number. The correspondence between graphs and exclusivity structures allows us compute the independence number of the graphs $H(K_{n})$, where $K_{n}$ is the complete graph, by taking advantage of well-known techniques used in quantum non-locality.

Since two graphs yield the same (up to isomorphism) graph in our construction if and only if they are equivalent under local complementation, the construction can be interpreted as a method to represent local complementation orbits. Somehow this is in analogy with the notion of a two-graph, a well-studied mathematical object which represents equivalence classes under the operation of switching (see Ch. 11 of \cite{GR01}).

Our work is innovative with respect to \cite{cabello10} and \cite{gunhe04} in many ways. First, we present a characterization of local complementation orbits, a result of pure combinatorial nature. Our representative graphs are obtained via a construction inspired by quantum mechanics. We find tools to analyze the properties of such graphs in \cite{gunhe04}. Second, by using the results in \cite{cabello10}, we discover that local complementation orbits are naturally associated to Bell inequalities. We improve the mathematical representation of such Bell inequalities by writing down an explicit operational form, namely, a \emph{pseudo-telepathy game}. The form that we have introduced is often easier to use in both theoretical purposes and the design of laboratory experiments. Third, we introduce a novel connection between the results presented in \cite{cabello10} and \cite{gunhe04}, and zero-error information theory. Such a connection is also interesting on its own, since it provides a way to construct families of channels with a separation between classical and quantum capacities starting from \emph{any} graph.

The remainder of the work is organized as follows. The next section introduces the required terminology and notions: the language of graph theory, non-locality, and channel capacities. The construction is described in Section 3. Section 4 discusses the relevant graph-theoretic parameters. Section 5 contains examples. We highlight that physical arguments can be useful to consider difficult tasks such as computing the independence number. Section 6 is devoted to zero-error capacities. We show that our construction produces infinite families of graphs for which the use of entanglement gives the maximum possible zero-error capacity. Section 7 classifies graphs (or local complementation orbits) according to the objects obtained via the construction. We point out a link with Boolean functions and propose a conjecture about connectedness of the graph representatives.


\section{Preliminaries}


\subsection{Graphs, graph parameters and graph states}


A (simple) \emph{graph} $G=(V,E)$ is an ordered pair:\ $V(G)$ is a set whose elements are
called \emph{vertices}; $E(G)\subset V(G)\times V(G)$ is a set whose
elements are called \emph{edges}. The set $E(G)$ does not contain an edge of
the form $\{i,i\}$, for every $i\in V(G)$. The vertices forming an edge are said
to be \emph{adjacent}. We denote by $\mathcal{N}(i)$ the \emph{%
neighbourhood} of the vertex $i$, \emph{i.e.}, $\mathcal{N}(i)=\{j\in
V\mid (i,j)\in E\}$. An \emph{independent set} in a graph $G$ is a set of
mutually non-adjacent vertices. The \emph{independence number} of a graph $G$,
denoted by $\alpha (G)$, is the size of the largest independent set of $G$.
A \emph{subgraph} $H=(V,E)$ of a graph $G=(V,E)$, is a graph such that $%
V(H)\subseteq V(G)$ and $E(H)\subseteq E(G)$.
An \emph{induced subgraph} $H=(V,E)$ of a graph $G=(V,E)$ with respect to
$V(H)\subseteq V(G)$ is a graph with vertex set $V(H)$ and edge set
$E(H) = \{ \{i,j\} \mid i,j\in V(H) \mbox{ and } \{i,j\} \in E(G) \} $.
A \emph{clique} in a graph $G$ is a subgraph whose vertices are all adjacent to each other.

An \emph{orthogonal representation} of $G$ is a map from $V(G)$
to $\mathbb{C}^k$ for some $k$, such that adjacent vertices are mapped to
orthogonal vectors. An orthogonal representation is \emph{faithful} when vertices
$u$ and $v$ are mapped to orthogonal vectors if and only if $\{u,v\} \in E(G)$.
The \emph{Lov\'asz number} \cite{Lovasz79} is defined as follows:
\begin{equation}
\vartheta (G)=\max \sum_{i=1}^{n}|\langle \psi |v_{i}\rangle |^{2},
\label{lovasz}
\end{equation}%
where the maximum is taken over all unit vectors $\psi$ and all orthogonal representations, $\{v_i\}$, of $G$.
The \emph{fractional packing number} is defined by the following linear program:
\begin{equation} \label{alphastar}
\alpha ^{\ast }(G,\Gamma)=\max \sum_{i\in V}w_{i},
\end{equation}%
where the maximum is taken over all $0 \leq w_i\leq 1$ under the restriction $\sum_{i \in C_j} w_i \leq 1$ \cite{SU97}, and for all cliques $C_j \in \Gamma$,
where $\Gamma$ is the set of all cliques of $G$.
In this paper, by fractional packing number, we mean $\alpha^* (G,\Gamma)$ and denote it as $\alpha ^{\ast }(G)$.

Given a graph $G=(V,E)$, the \emph{graph state} $|G\rangle $ (see, for example, \cite{hein03, sli01})
associated to $G$ is the unique $n$-qubit state such that
\begin{equation}
g_{i}|G\rangle =|G\rangle \text{ for }i=1,\ldots ,n, \label{graphdef}
\end{equation}%
where $g_{i}$ is the generator labeled by a vertex $i\in V$ of the
stabilizer group of $|G\rangle $. A generator for $i\in V(G)$ is defined as
\begin{equation}
g_{i}=X^{(i)}\bigotimes\nolimits_{j\in \mathcal{N}(i)}Z^{(j)},
\label{stabdef}
\end{equation}%
where $X^{(i)}$, $Y^{(i)}$, and $Z^{(i)}$ denote the Pauli matrices
(sometimes denoted as $\sigma _{x}$, $\sigma _{y}$, and $\sigma _{z}$)
acting on the $i$-th qubit. Therefore, $g_{i}$ can be obtained directly and univocally
from $G$.
The stabilizer group of the state $|G\rangle $ is the set $S$ of the
stabilizing operators $s_{j}$ of $|G\rangle $ defined by the product of any
number of generators $g_{i}$. Note that, for convenience, we shall remove
the identity element from $S$. Therefore, the set $S$ contains $2^{n}-1$
elements.

Given a graph $G = (V,E)$, the operation of \emph{local complementation} (LC) on $i\in V$ transforms $G$ into a graph $G^i$ on the same set of vertices. To obtain $G^i$, we replace the induced subgraph of $G$ on ${\mathcal{N}}(i)$ by its complement. It is easy to verify that $(G^i)^i = G$. The set of graphs is partitioned into \emph{LC orbits} (also known as \emph{Kotzig orbits}) by the repeated action of local complementation on each graph \cite{bouchet88}. The LC orbits are then equivalence classes.


\subsection{Non-locality}


We assume familiarity with the basics of quantum information theory. The reader can find a good introduction
in \cite[Chapter 2]{NielsenChuang}.
A \emph{non-local game} is an experimental setup between a referee and two players, Alice and Bob. (It can also be defined with more players, but we do not consider this case here.)
The game is not adversarial, but the players collaborate with each other. They are allowed to arrange a strategy beforehand, but they are not allowed to communicate during the game. The referee sends Alice an input $x\in X$ and sends Bob an input $y\in Y$,
according to a fixed and known probability distribution $\pi$ on $X \times Y$. Alice and Bob answer with $a\in A$ and $b\in B$ respectively, and the referee declares the outcome of the game according
to a verification function $V: A\times B\times X\times Y \rightarrow \{$win $=1$, lose $=0\}$.
So, the non-local game is completely specified by the sets $X,Y,A,B$, the distribution $\pi$, and the verification function $V$.

A \emph{classical strategy} is \emph{w.l.o.g.}\ a pair of functions $s_A : X\rightarrow A$ and $s_B: Y\rightarrow B$ for Alice and Bob, respectively. A \emph{quantum strategy} consists of a shared bipartite entangled state $\ket{\psi}$ and POVMs
$\{P^x_a\}_{a\in A},$ for every $x\in X$ for Alice and $\{P^y_b\}_{b\in B},$ for every $y\in Y$ for Bob. On input $x$, Alice uses the positive operator valued measurement (POVM) $\{P^x_a\}_{a\in A}$ to measure her part of the entangled state and Bob does similarly on his input $y$. Alice (resp.\ Bob) answers with $a$ (resp.\ $b$) corresponding to the obtained measurement outcome. Therefore, the probability to output $a,b$ given $x,y$ is $\Pr(a,b|x,y) = \bra{\psi} P^x_a \otimes P^y_b \ket{\psi} $. The classical and quantum values (or winning probabilities) for the game are:
\begin{subequations}
\begin{align}
\omega_c &= \max_{s_A, s_B} \sum_{x,y,a,b} \pi(x,y) V(s_A(x),s_B(y),x,y),\\
\omega_q &= \max_{\ket{\psi}, \{P^x_a\}, \{P^y_b\} } \sum_{x,y,a,b} \pi(x,y) V(a ,b,x,y) \bra{\psi} P^x_a \otimes P^y_b \ket{\psi}.
\end{align}
\end{subequations}

A \emph{Bell inequality} for a non-local game is a statement of the form $ \omega_c \leq t $ for $t \in [0,1]$. It is violated by quantum mechanics if $ \omega_q > t$.
A non-local game is called a \emph{pseudo-telepathy} game if $\omega_c < \omega_q = 1 $, \emph{i.e.}\ quantum players win with certainty, while classical players have nonzero probability to lose.

Non-local games are a special form of Bell experiment. In general, a \emph{Bell operator} $\mathcal{B}$ is a linear combination of observables and a Bell inequality is a statement of the form
\begin{equation}
\max | \langle\mathcal{B}\rangle | \leq t,
\end{equation}
where the maximum runs over classical states. A quantum state is said to \emph{violate} the Bell inequality if $|\langle\mathcal{B}\rangle| > t$.


\subsection{Channel capacity}


Zero-error information theory was initiated in \cite{shannon}; a review is \cite{korner}. A \emph{classical channel} $\mathcal{C}$ with input set $X$ and output
set $Y$ is specified by a conditional probability distribution $\mathcal{C}(y|x)$, the probability to produce output $y$ upon input $x$. (Precisely this is a discrete, memoryless, stationary channel.)
Two inputs $x,x'\in X$ are \emph{confusable} if
there exists $y\in Y$ such that $\mathcal{C}(y|x)>0$ and $\mathcal{C}(y|x')>0.$
We then define the \emph{confusability graph of channel $\mathcal{C}$}, $G(\mathcal{C})$, as
the graph with vertex set $X$ and edge set $\{(x,x'):x,x'\mbox{ are distinct and confusable}\}$.

The \emph{one-shot zero-error capacity} of $\mathcal{C}$, $c_0(\mathcal{C})$,
is the size of a largest set of non-confusable inputs. This is just the independence number $\alpha(G(\mathcal{C}))$ of the confusability
graph. In the entanglement-assisted setting, the sender (Alice) and receiver (Bob) share an entangled state $\rho$ and can perform local quantum measurements on their part of $\rho$.

The general form of an entanglement-assisted protocol used by Alice to send one out of $q$ messages to
Bob with a single use of the classical channel $\mathcal{C}$ can be described as follows (also see \cite{CLMW10}).
For each message $m\in [q]$, Alice has a POVM $\mathcal{E}^m = \{E_{1}^{m},\dots,E_{|X|}^{m}\}$ with $|X|$ outputs.
To send message $m$, she measures her subsystem
using $\mathcal{E}^m$ and sends through the channel the
observed $x\in X$.
Bob receives some $y\in Y$ with $\mathcal{C}(y|x)>0$.
If the right condition holds (as we will explain below), Bob can recover $m$ with certainty
using a projective measurement on his subsystem.

It is not hard to state a necessary and sufficient condition for the success of the protocol.
If Alice gets outcome $x\in X$ upon measuring $\mathcal{E}^m$,
Bob's part of the entangled state collapses to
$\beta_{x}^{m}=\Tr_{A}((E_{x}^{m}\otimes I)\rho)$. Given the channel's output $y$, Bob can recover $m$ if and only if
\begin{equation}
\forall m\neq m',\forall\mbox{ confusable }x,x'\ \Tr(\beta_{x}^{m}\beta_{x'}^{m'})=0.\label{eq:condition}
\end{equation}
Bob can recover the message with a projective measurement on the mutually orthogonal supports of
\begin{equation}
\sum_{x \;:\; \mathcal{C}(y|x)>0} \beta_{x}^{m},
\end{equation}
for all messages $m$.
In such a case we say that, assisted by the entangled state $\rho$, Alice can use the POVMs
$\mathcal{E}^1,\dotsc,\mathcal{E}^q$ as her strategy for sending one out of $q$ messages with a single use of $\mathcal{C}$.

The \textit{entanglement-assisted one-shot zero-error channel capacity }of $\mathcal{C}$, $c_0^*(\mathcal{C})$, is the maximum integer $q$ such that
there exists a protocol for which condition \eqref{eq:condition} holds.

We are now ready to outline the setting where Alice and Bob share a maximally entangled state in the above protocol. We will refer to this particular case later in section \ref{sec:capacity}.
Let the (canonical) maximally entangled state of local dimension $n$ be defined as follows:
\begin{equation}
\ket{\Psi} := \frac{1}{\sqrt n} \sum_{i\in [n]} \ket{i}\ket{i},
\end{equation}
where $\{\ket{i}\}_{i\in [n]}$ is the standard basis of $\C^n$.
When Alice and Bob share a maximally entangled state and Alice performs a projective measurement observing $P_a$,
Bob's part of the state collapses to $\Tr_{A}((P_a\otimes I)\ketbra{\Psi}{\Psi}) = P_a^\top /n$. This implies that Bob can distinguish between $P_a$ and $P_b$ perfectly if and only if $\Tr(P_a P_b)=0$. Therefore, if Alice uses projective measurement $\{P^m_x\}_{x\in X}$ for message $m$ and players share a maximally entangled state, then Condition
\eqref{eq:condition} is true if and only if
\begin{equation}
\forall m\neq m',\forall\mbox{ confusable }x,x'\ \Tr(P_{x}^{m}P_{x'}^{m'})=0.\label{eq:condition2}
\end{equation}

Considering more than a single use of the channel, one can define the asymptotic zero-error channel capacity $\Theta(\mathcal{C})$ and the asymptotic entanglement-assisted zero-error channel capacity $\Theta^*(\mathcal{C})$ by
$ \displaystyle \Theta(\mathcal{C}) :=
    \lim_{k\rightarrow\infty} (c_0(\mathcal{C}^{\otimes k}))^{1/k}
 $ and $
\displaystyle
 \Theta^*(\mathcal{C}) :=
    \lim_{k\rightarrow\infty} (c_0^*(\mathcal{C}^{\otimes k}))^{1/k}.
$

Since $c_0^*(\mathcal{C})$ depends solely on the confusability graph $G(\mathcal{C})$ \cite{CLMW10}, we can talk about $c_0^*(G)$ for a graph $G$, meaning the entanglement-assisted one-shot zero-error capacity of a channel with confusability graph $G$. Similarly we can talk about quantities $c_0(G), \Theta(G),$
and $\Theta^*(G)$.


\section{Construction}\label{con}


Let $G$ be a graph on $n$ vertices and consider the $n$-qubit graph state $|G\rangle $.
Let $S$ be the stabilizer group of $G$.
For each $s_j \in S$, with $s_j = \bigotimes_{k=1}^n O^{(k)}$,
let $w_j = |\{ O^{(k)}: O^{(k)}\neq I \}|$ be the \emph{weight of} $s_j$.
Let $S_j = \{S_{(i,j)}: i=1,2,\ldots,2^{w_j-1}\}$ be the set of the \emph{events} of $s_j$,
\emph{i.e.} the measurement outcomes that occur with non-zero probability
when the system is in state $|G\rangle$ and the stabilizing operators $s_{j}$ are measured
with single-qubit measurements.
The set of all events is $\mathcal{S} = \bigcup_{j=1,2,\ldots,2^{n}-1} S_j$.
Two events are \emph{exclusive} if there exists a $k \in \{1,\ldots,n\}$
for which the same single-qubit measurement gives a different outcome.

A graph representing a Kotzig orbit can be naturally defined as follows:

\begin{definition}\label{def:H}
Let $G$ be a graph. Let $S$ be the stabilizer group of the graph state of $G$. We denote by $H(G)$ the graph whose vertices are the events in $S$ and the edges are all the pairs of exclusive events.
\end{definition}

We give an example for events and exclusiveness.
Let $n=3$ and $s_{2}=ZXZ$ (we omit the superscripts for simplicity).
This means that $ZXZ|G\rangle
=|G\rangle $, \emph{i.e.}, if the system is prepared in $|G\rangle $ and $%
s_{2}$ is measured by measuring $Z$ on the first qubit (with possible
results $-1$ or $1$), $X$ on the second qubit, and $Z$ on the third qubit,
then the product of the three results must be $1$. Therefore,
$S_{2}= \{zxz,z\underline{x}\underline{z}, \underline{z}x\underline{z}, \underline{z}\underline{x} z\}$,
where hereafter $z\underline{x}\underline{z}$ denotes the event \textquotedblleft the result 1
is obtained when $Z$ is measured on qubit 1, the result $-1$ is obtained
when $X$ is measured on qubit 2, and the result $-1$ is obtained when $Z$ is
measured on qubit 3\textquotedblright . As another example: if $n=2$ and $s_{1}=XZ
$, then $S_{1}=\{zx, \underline{z}\underline{x}\}$.

We now give an example of a graph representing a Kotzig orbit.
Let us consider $P_{3}=(\{1,2,3\},\{\{1,2\},\{2,3%
\}\})$, the path on three vertices. We construct $H(P_{3})$. The stabilizer
group $S$ (minus the identity)\ has the following elements: $s_{1}=g_{1}=XZI$%
, $s_{2}=g_{2}=ZXZ$, $s_{3}=g_{3}=IZX$, $s_{4}=g_{1}g_{2}=YYZ$, $%
s_{5}=g_{1}g_{3}=XIX$, $s_{6}=g_{2}g_{3}=ZYY$, and $%
s_{7}=g_{1}g_{2}g_{3}=-YXY$. For all $j=1,\ldots ,2^{3}-1$, obtain all
possible events (\emph{i.e.}, those which can happen with non-zero
probability) when three qubits are prepared in the state $|G\rangle $ and
three parties measure the observables corresponding to $s_{j}$. For
instance, when $j=1$, Alice measures $X^{(1)}$, Bob measures $Z^{(2)}$, and
Charlie does not perform any measurement . Since the three qubits are in
state $|G\rangle $, there are only two possible outcomes: Alice obtains $%
X^{(1)}=+1$ and Bob obtains $Z^{(2)}=+1$, denoted as $xzI$;
or Alice obtains $X^{(1)}=-1$ and Bob obtains $Z^{(2)}=-1$, denoted as $\underline{x}\underline{z}I$.
For $j=2$, the only events that can occur
are $zxz$, $z\underline{x}\underline{z}$, $\underline{z}x\underline{z}$, and
$\underline{z}\underline{x}z$. The other events for the remaining $j$'s are
obtained in a similar way.
Now, let us construct the graph $H(P_{3})$: the vertices represent possible
events; two vertices are adjacent if and only events are exclusive
(\emph{e.g.}, $xzI$ and $\underline{x}I \underline{x}$).
Notice that each $s_{j}$ of weight $w_j$ generates $2^{w_j-1}$ vertices.
A drawing of $H(P_{3})$ is in Fig. \ref%
{H1}.


\begin{figure}[t]
\centerline{\includegraphics[width=6.6cm]{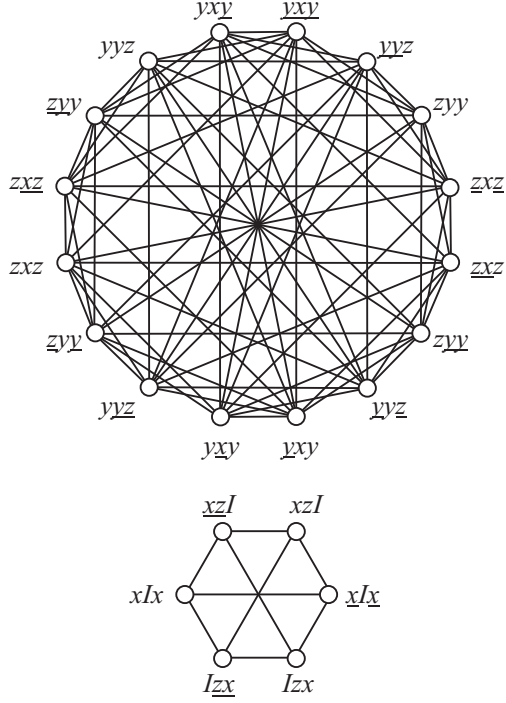}}
\caption{The graph $H(P_3)$ associated to the path on three vertices, $P_3$, consists of two
connected components: the upper component in the drawing is $CS$, the complement of the Shrikhande
graph \cite{Shrikhande59}; the lower component is $Ci_8(1,3)$, the 6-vertex (1,3)-circulant graph. We have $%
\protect\alpha(CS)=3$, $\protect\vartheta(CS)=\protect\alpha^* (CS)=4$, $%
\protect\alpha[Ci_8(1,3)]=\protect\vartheta[Ci_8(1,3)]=\protect\alpha^*
[Ci_8(1,3)]=3$. Therefore, $\protect\alpha(H(P_3))=6$, while $\protect\vartheta%
(H(P_3))=\protect\alpha^* ((P_3))=7$.}
\label{H1}
\end{figure}


Each $H(G)$ can be interpreted as in \cite{cabello10}. Every graph is in fact associated to an NC inequality and is constructed by expressing the linear combination of joint probabilities of events in the NC inequality as a {\em sum} $S$. For a graph in \cite{cabello10}, an event in $S$ is represented by a vertex and exclusive events are represented by edges. Constructing such a graph from the inequality is straightforward, when the absolute values of the coefficients in the linear combination are natural numbers (which, to our knowledge, is always the case for all relevant NC inequalities). As already mentioned in the introduction, this graph-theoretic framework can be used to single out games with \emph{ad hoc} quantum advantage and quantum correlations with \emph{ad hoc} degree of contextuality (see \cite{NDSC12, ADLPBA11}, respectively).


\subsection{Local complementation orbits}
\label{Kotzig}


If we apply the method to graphs $G$ and $G^{\prime }$ in the
same orbit under local complementation \cite{hein03, CLMP09} then \emph{we
obtain the same graph} $H$. The reason is that the graph states $|G\rangle$
and $|G^{\prime }\rangle $ share the same set of perfect correlations (up to
relabeling), so also share the same graph in which all possible
exclusive events are adjacent.

This paper constructs $H(G)$ from $G$, as described earlier, where each of
the $2^n-1$ operators, $s_j$, generated by $G$, in turn generates a clique
of size $2^{w_j-1}$ in $H$, where $w_j$ is the weight of operator $s_j$.
In this section, we present a
classification of all $H(G)$ from all graphs $G$ for $n < 7$. This classification
is greatly simplified by the fact that if two graphs, $G$ and $G^{\prime }$,
are in the same local complementation (LC) orbit, then $H(G) = H(G^{\prime
}) $. So we need only classify for one representative from each orbit.
 A choice of representatives
for $n=2,\ldots ,6$, for connected graphs only, are listed in the
second column of table \ref{BigClass}.

The action of local complementation on vertex $v$ of graph $G$ to yield
graph $G^{\prime }$ can be realised, in the context of graph states, by a
specific local unitary action:
\begin{equation}
\begin{array}{rcl}
G^{\prime } & = & G^v \\
\Updownarrow & & \Updownarrow \\
|G^{\prime }\rangle & = & \omega^7 T^{(v)}N^{(v)}\prod_{i \in {\mathcal{N}}_v}
T^{(i)} |G\rangle,%
\end{array}%
\label{LUMap}
\end{equation}
where $N = \frac{1}{\sqrt{2}}\left (
\begin{array}{rr}
1 & i \\
1 & -i%
\end{array}
\right )$, $T = \left (
\begin{array}{rr}
1 & 0 \\
0 & i%
\end{array}
\right )$, and $\omega = e^{\frac{7\pi i}{4}}$.

\vspace{2mm}

Here, $S = \{s_j, j = 1 \ldots 2^n-1\}$ is the stabilizer group associated with the graph $G$,
where we omit $s_0$ for convenience.
Similarly, let $S' = \{s'_j, j = 1 \ldots 2^n-1\}$ be the stabilizer group associated with $G' = G^v$.
We show how to obtain $S'$ from $S$. Let $s_j = (-1)^{c_j}\prod_{k=1}^n O_j^{(k)}$, where
$O_j \in \{I,X,Z,Y\}$, and $c_j \in \{0,1\}$.

\vspace{2mm}

Define the mapping ${\cal{L}}_v : \{I,X,Z,Y\} \rightarrow \{I,X,Z,Y\}$ as follows:
\begin{equation}
\begin{array}{lll}
{\cal{L}}_v : & X^{(k)} \rightarrow X^{(k)}, Z^{(k)} \rightarrow Y^{(k)}, Y^{(k)} \rightarrow Z^{(k)}, & k = v \\
       & X^{(k)} \rightarrow Y^{(k)}, Z^{(k)} \rightarrow Z^{(k)}, Y^{(k)} \rightarrow X^{(k)}, & k \in {\mathcal{N}}_v \\
       & I \rightarrow I, X \rightarrow X, Z \rightarrow Z, Y \rightarrow Y, & otherwise.
\end{array}
\end{equation}

Moreover, define $y_v(s_j)
 = |\{k \hspace{2mm} | \hspace{2mm} O_j^{(k)} = Y^{(k)} \mbox{ and } k \in v \cup {\mathcal{N}}_v \}|$. In
words, $y_v(s_j)$ is the total number of $Y$ matrices at tensor positions $v \cup {\mathcal{N}}_v$
of $s_j$.

\begin{lemma}
\label{LCPerm}
The action of LC at vertex $v$ of $G$ maps $G$ to $G'$ and $S$ to $S'$, where
\begin{equation}
 s_j' = (-1)^{c_j + y_v(s_j)} \prod_{k=1}^n {\cal{L}}_v(O_j^{(k)}), \hspace{2mm} j = 0 \ldots 2^n-1.
\end{equation}
This action is a permutation, $(I)(X)(ZY)$, of the Pauli matrices at tensor position $v$ of each $s_j$
and a permutation, $(I)(Z)(XY)$, of the Pauli matrices at tensor positions in ${\mathcal{N}}_v$
of each $s_j$, followed by a global multiplication by $(-1)^{y_v(s_j)}$.
\end{lemma}
For example, consider the graph $G$ with two edges $\{1,2\}$ and $\{1,3\}$. Then $G' = G^1$ is the graph with edges $\{1,2\},\{1,3\},\{2,3\}$ (so the star $ST_3$ and the complete graph $K_3$
are in the same Kotzig orbit). We have that $S(G) = \{XZZ,ZXI,YYZ,ZI X,YZY,IXX,-XYY\}$ and that \newline
$S(G^1) = S' = \{XZZ,ZXZ,YYI ,ZZX,YI Y,I YY,-XXX\}$. For example $s_j = YYZ$ is mapped to
$s_j' = (-1)^{y_1(s_j)}ZXZ = ZXZ$, where $v = 1$, ${\mathcal{N}}_v = \{2,3\}$, and $y_1(s_j) = 2$
as $Y$ occurs
at tensor positions $1$ and $2$ of $s_j$, where $1,2 \in v \cup {\mathcal{N}}_v$.
\\

\begin{proof}
We use (\ref{LUMap}).
For vertex $v$ we replace $U$ with $\omega^7 TNUN^{\dag}T^{\dag} \omega$,
for each of $U \in \{X,Z,Y\}$ to obtain $\{X,Y,-Z\}$. Similarly, for vertices
in ${\mathcal{N}}_v$ we replace $U$ with $\omega^7 TUT^{\dag} \omega$, for
each of $U \in \{X,Z,Y\}$ to obtain $\{Y,Z,-X\}$.
\end{proof}

\begin{theorem}
\label{LCInvariance}Let $G^{L}$ be the Kotzig orbit of graphs generated by the
action of successive local complementation on $G$. Then
\begin{equation}
H(G^{\prime })=H(G^{\prime \prime }),\hspace{5mm}\forall G^{\prime
},G^{\prime \prime} \in G^L.
\end{equation}
\end{theorem}

\begin{proof}
Every vertex in $H(G)$ represents a measurement, $s_j$, of $|G\rangle$, combined with a certain
measurement result, as specified by the bars under $x$, $y$, and $z$, as appropriate.
This measurement is equivalent to a measurement, $s_j'$ of $|G' \rangle$, where $s_j'$,
$|G' \rangle$, and the new measurement results are obtained from $s_j$ and
$|G\rangle$ by the same local unitary transform, namely the transform in (\ref{LUMap}). Since the
two measurement scenarios are equivalent, then the edge relationship between vertices in
$H(G)$ is preserved in $H(G')$, \emph{i.e.}, $H(G) = H(G')$. The theorem is then extended to any two
$G^{\prime },G^{\prime \prime } \in G^L$ as $G^{\prime \prime }$ can be obtained from
$G^{\prime }$ by a series of local complementations.
\end{proof}

\vspace{3mm}

Proposition \ref{LCInvariance} implies that we only have to classify for one
representative member, $G$, (arbitrarily chosen) of each Kotzig orbit, $G^L$, of
$n$-vertex graphs. Table \ref{BigClass} classifies, computationally, graphs $H(G) = (V^H,E^H)$
for $n = 2, \ldots 6$ from $G = (V,E)$, where $|G^L|$ is the size of the Kotzig
orbit of $G$ up to re-labeling (graph isomorphism).


\section{Parameters}


\begin{theorem}
\label{thm:Hparams}Let $H$ be a graph representing a Kotzig orbit. Then, for $n > 2$,
\begin{equation}
\alpha (H)<\vartheta (H)=\alpha ^{\ast }(H)=2^{n}-1.
\end{equation}
\end{theorem}

We firstly give an intuition of the statement, explaining how the Theorem can be seen as a consequence of the results in \cite{cabello10,gunhe04}. A formal and stand-alone proof will follow later in this section. 
Let $|G\rangle $ be the graph state with
corresponding graph $G$. It was shown in \cite{gunhe04} that the sum of the
elements of the stabilizer group of $|G\rangle $, $\sum_{j=1}^{2^{n}-1}s_{j}$
is a Bell operator such that $\sum_{1}^{2^{n}-1}s_{j}|G\rangle
=(2^{n}-1)|G\rangle $ and $\max |\langle S \rangle|<2^{n}-1$ when
restricting to classical states (where $S$ is the stabilizer group defined earlier). In
words, the graph state $|G\rangle $ violates the corresponding Bell
inequality up to its algebraic maximum. This fact together with
\cite[Equation 6]{cabello10} enforces that $\alpha (H)<\vartheta (H)$ and $%
\vartheta (H)=\alpha ^{\ast }(H)=2^{n}-1$. The construction in Definition \ref{def:H}
simply transforms the Bell operator, originally written as a sum of mean values,
into a sum of probabilities of events, in order to construct the graph
associated with the exclusivity structure.

The statement, therefore, combines known facts from quantum information in a novel way in order to prove a purely graph theoretical result.

The proof of Theorem \ref{thm:Hparams} requires the following definition:
\begin{definition}[Canonical orthogonal representation]
\label{def:orth}Let $H=(V,E)$ be a graph representing a Kotzig orbit. Let
$ S_{(i,j)} = \left( s^{(1)}_{(i,j)},s^{(2)}_{(i,j)},\ldots,s^{(n)}_{(i,j)}\right) $ be the event
at vertex $(i,j)\in V(H)$, where $i = 1 \ldots 2^n-1$, $j = 1 \ldots 2^{w_i - 1}$, and
$s^{(k)}_{(i,j)}\in \{I,x,\underline{x},y,%
\underline{y},z,\underline{z}\}$, for each $k=1,2,\ldots,n$. Let $%
|s^{(k)}_{(i,j)}\rangle $ be defined as follows:%
\begin{equation}
\begin{tabular}{l|lllllll}
$s^{(k)}_{(i,j)}$ & $x$ & $\underline{x}$ & $y$ & $\underline{y}$ & $z$ & $%
\underline{z}$ & $I$ \\ \hline
$|s^{(k)}_{(i,j)}\rangle $ & $|+\rangle $ & $|-\rangle $ & $|y_{+}\rangle $ & $%
|y_{-}\rangle $ & $|0\rangle $ & $|1\rangle $ & $|\psi \rangle $%
\end{tabular}%
.
\end{equation}%
Here, $|\psi \rangle $ is an arbitrary ray in $\mathbb{C}^{2}$ and $|y_{+}\rangle, |y_{-}\rangle$
are the eigenvectors of the Pauli matrix $Y$ with eigenvalue $+1$ and $-1$, respectively.
The \emph{canonical orthogonal representation} of $H$ is the set of vectors
$\{|s_{(i,j)}\rangle :=|s^{(1)}_{(i,j)}\rangle \otimes |s^{(2)}_{(i,j)}\rangle
\otimes \cdots \otimes |s^{(n)}_{(i,j)}\rangle :(i,j)\in V(H)\}$.
\end{definition}

For example, in $H(P_{3})$ (see Fig. \ref{H1}), the element of the canonical
orthogonal representation of the vertex labeled by $\underline{x}I\underline{%
x}$ is $|-\rangle \otimes |\psi \rangle \otimes |- \rangle $. Notice that
if $|\psi \rangle $ is chosen to be non-orthogonal to any of the vectors $%
|+\rangle ,|-\rangle ,\ldots ,|1\rangle $ then the representation is faithful.
\bigskip

\begin{proof}[Proof of Theorem \protect\ref{thm:Hparams}]
Let $H=H(G)$ be a graph representing a Kotzig orbit of a graph $G$. The
proof is structured in three parts: \emph{(1)} we prove that $\vartheta (H)\geq
2^{n}-1$; \emph{(2)} we prove that $\alpha ^{\ast }(H)\leq 2^{n}-1$; \emph{%
(3)}\ finally, we prove that $\alpha (H)<2^{n}-1$. The first two parts
together prove that $\vartheta (H)=\alpha ^{\ast }(H)=2^{n}-1$, since $\vartheta
(G)\leq \alpha ^{\ast }(G)$, for any graph $G$ (see, \emph{e.g.}, \cite%
{cabello10}). We begin with the first part:

\emph{(1) }It follows directly from Eq. \eqref{graphdef} that $%
\sum_{i=1}^{2^{n}-1}\langle G|s_{i}|G\rangle =2^{n}-1$. We know that the
eigenvectors with eigenvalue $+1$ of each operator $s_{i}$ are in one-to-one
correspondence with the vertices of a clique in $H$: $|s_{(i,1)}\rangle
,|s_{(i,2)}\rangle,\ldots,|s_{(i,2^{w_{i}-1})}\rangle $. These are elements of
the canonical orthogonal representation of $H$. From the definition of the
stabilizer group, for all $s_{i}\in S$ and for all eigenvectors $|\underline{%
s}^{(i,j)}\rangle $ ($j=1,2,\ldots,2^{w_{i}-1}$) with eigenvalue $-1$, we have $%
\langle \underline{s}^{(i,j)}|G\rangle =0$, because $|G\rangle $ is in the $%
+1$ eigenspace. Now, let $s_{i}=\sum_{j}\lambda _{ij}|s_{(i,j)}\rangle
\langle s_{(i,j)}|$ be an Hermitian eigendecomposition of $s_{i}$. Thus,
\begin{eqnarray}
2^{n}-1 &=&\sum_{i=1}^{2^{n}-1}\langle G|s_{i}|G\rangle \\
&=&\sum_{i=1}^{2^{n}-1}\sum_{j}\lambda _{ij}\langle G|s_{(i,j)}\rangle
\langle s_{(i,j)}|G\rangle \\
&=&\sum_{i=1}^{2^{n}-1}\sum_{j:\lambda _{ij}=1}\langle G|s_{(i,j)}\rangle
\langle s_{(i,j)}|G\rangle \\
&=&\sum_{i=1}^{2^{n}-1}\sum_{j:\lambda _{ij}=1}|\langle G|s_{(i,j)}\rangle
|^{2} \\
&\leq &\vartheta (H),
\end{eqnarray}%
where the inequality in the last line follows because a canonical
orthogonal representation of $H$ together with the state $|G\rangle $
represents a feasible solution for the semidefinite formulation of the Lov\'{a}sz
number in Eq. \eqref{lovasz}.

\emph{(2) }From the linear programming formulation of the fractional packing
number (see Eq. \eqref{alphastar}), it is easy to see that a partition of
the set of vertices into $k$ cliques gives an upper bound to $\alpha ^{\ast}(H)$.
To see this choose one vertex, say $i$, per clique and set its weight
$w_{i}=1$. We get a partition of $H$ into $2^{n}-1$ cliques if we consider the
events associated with each $s_{i}$.

\emph{(3)\ }We use an argument very similar to \cite[Lemma 1 and
Theorem 1]{gunhe04}. If the number of vertices of $G$ is two then the result
does not hold as $\al(H) = 2^2-1$ (by direct calculation).
Each connected graph with more than two
vertices has a subgraph with three vertices. For each of those we can see
(also by direct calculation; see Table \ref{BigClass} of section \ref{Kotzig})
that $\alpha (H)<7$.
Therefore, we just need to show that if $G^{\prime }$ is a subgraph of $G$
with $n^{\prime }$ vertices and $\alpha (H^{\prime })<2^{n^{\prime }}-1$, where $%
H^{\prime }$ is the representative of the Kotzig orbit of $G^{\prime }$, then $%
\alpha (H)<2^{n}-1$ for $n > 2$. Notice that $S^{\prime }$, the stabilizer group of $%
G^{\prime }$, is a subset of $S$. Therefore, in the graph $H$ we find
cliques associated with $S^{\prime }$, but containing slightly different
events. For each $s_{i}^{\prime }\in S^{\prime }$, the corresponding $%
s_{i}\in S$ has the same structure, with eventually some additional $Z$
operators. Let $\tilde{H}$ be the subgraph of $H$ induced by the vertices in
cliques associated with the elements of $S^{\prime }$. We need to show that if
in $H^{\prime }$ there is no vertex per clique to form a maximal independent
set then neither are there in $\tilde{H}$. Therefore, $\alpha (H)<2^{n}-1$.
Towards a contradiction, suppose there is an independent set $L$ of
$\tilde{H}$ such that $|L|=2^{n^{\prime }}-1$. We distinguish two cases:

\begin{itemize}
\item If the events at the vertices in $L$ do not have any $\underline{z}$
element then we can map them to an independent set in $H^{\prime }$ of size
$2^{n^{\prime }}-1,$ just by ignoring the additional $Z$ operators. This
contradicts the hypothesis that $\alpha (H^{\prime })<2^{n^{\prime }}-1$.

\item If the events at the vertices in $L$ do have $\underline{z}$ elements
then we can find another independent set $J$ with the same cardinality such
that the events at its vertices do not have any $\underline{z}$ element. We
can find $J$ as follows. It is easy to check that an operator $s_{i}$ has
the form $O^{(1)}\cdots Z^{(\ell) }\cdots O^{(n)}$ if and only if it has an odd number of
$X^{(k)}$ and $Y^{(k)}$, with $\{\ell ,k\}\in E(H)$. Therefore, complementing
$\underline{z^{(\ell)}}$ and all occurrences of $X^{(k)}$ and $Y^{(k)}$ in
the events at the vertices of the independent set $L$, we obtain the events
in $J$ with the desired properties, and so we are back to the previous case.
\end{itemize}
\end{proof}

\bigskip

The Bell inequalities described by the graph $H$ are exactly the same as in
\cite{gunhe04}, but in the form of a pseudo-telepathy game.

\begin{definition}
For any graph $G$ on $n$ vertices, let us define an $n$-player game for $G$
as follows. The \emph{input set }for each player is $\{X,Y,Z,I\}$ and the
\emph{output set} is $\{+1,-1\}$. The set of valid inputs is the set of
elements of the stabilizer group of $|G\rangle $. The players win on input
$s_{i}\in S$ if and only if the sign of the product of their outputs equals
the sign of $s_{i}$.
\end{definition}

\begin{corollary}
The graph game for $G$ is a pseudo-telepathy game.
\end{corollary}

\begin{proof}
It is easy to see that if the players share the graph state $|G\rangle $ and
each player performs the measurement corresponding to her input, then they
always win. On the other hand, we show that a classical strategy for the game
can be used to construct an independent set of $H=H(G)$ and \emph{viz}.

We now consider the first direction. If there exists a strategy which
answers correctly to $k$ questions then there exists an independent set with
$k$ elements. A classical strategy is \emph{w.l.o.g.}\ a set of functions
for each player from the input set to the output set. Therefore, for all
the winning inputs $s_{i}$, there will be a single output $(a_{1},\dots,a_{n})$,
corresponding to a vertex of $H$. It is easy to verify that there
cannot be an edge between any pair of these vertices. Since the strategy
wins on $k$ input pairs, the independent set has $k$ elements.

For the other direction, we show that if there exists an independent set $L$
of $H$ having size $k$, then there exists a strategy for the game on $G$
that answers correctly to at least $k$ of the $2^{n}-1$ questions. By the
structure of $H$, the independent set $L$ cannot contain vertices $i,j$ such
that, for the same input $x$, $a_{\ell }^{(i)}\neq a_{\ell }^{(j)}$ for some
$\ell \in \{1,\dots n\}$. Hence, we have the following strategy: on input $x$,
each player outputs the unique $a$ determined by the vertices in the
independent set. The size $k$ of the independent set implies that the
players answer correctly to at least $k$ input pairs.
\end{proof}

\section{Examples}
\label{Examples}

\begin{proposition}
Let $K_{n}$ be the complete graph on $n$ vertices. Then,%
\begin{equation}
\alpha (H(K_{n}))=%
\begin{cases}
2^{\frac{n-3}{2}}+3\cdot 2^{n-2}-1, & \text{for $n$ odd;} \\
2^{\frac{n}{2}-1}-2^{n-2}+2^{n}-1, & \text{for $n$ even.}%
\end{cases}
\label{al}
\end{equation}
\end{proposition}

\begin{proof}
As said before, $H(G)$ can be associated to the Bell inequality in which the
Bell operator is the sum of all stabilizer operators of $|G\rangle $ \cite%
{gunhe04}. The first observation is that the graph state associated to $%
G=K_{n}$ is the $n$-qubit Greenberger-Horne-Zeilinger (GHZ) state \cite%
{hein03}. The second observation is that, in that case, the Bell inequality
corresponding to $H(G)$ is the sum of a well-known Bell inequality maximally
violated by the $n$-qubit GHZ state \cite{mermin90} plus a trivial Bell
inequality not violated by quantum mechanics \cite{cabello08}. For example,
for $n=3$, the Bell inequality corresponding to $H(K_{3})$ is $\beta
_{3}=\mu _{3}+\tau _{3}\leq 2+3$, where
\begin{subequations}
\begin{align}
\mu _{3}& =\langle XZZ\rangle +\langle ZXZ\rangle +\langle ZZX\rangle
-\langle XXX\rangle, \\
\tau _{3}& =\langle YYI\rangle +\langle YIY\rangle +\langle IYY\rangle;
\end{align}
\end{subequations}
recall that $\langle XZZ\rangle $ denotes the mean value of the product of
the outcomes of the measurement of $X$ on qubit $1$, $Z$ on qubit $2$, and $%
Z $ on qubit $3$. The inequality $\mu _{3}\leq 2$ is the well-known Bell
inequality introduced in \cite{mermin90}, while $\tau _{3}\leq 3$ is a
trivial inequality not violated by quantum mechanics \cite{gunhe04,
cabello08}. The sum $\beta _{3}$ has $7$ terms, with four terms generating cliques of size $2^3 - 1$ and the other three terms generating cliques of size $2^2 - 1$. Then, we can see that $\alpha (H(K_{3}))$ is equal
to the maximum number of quantum predictions that a deterministic local
theory can simultaneously satisfy. By quantum predictions we mean: $XZZ=1$, $%
ZXZ=1$, $ZZX=1$, $XXX=-1$, $IYY=1$, $YIY=1$, and $IYY=1$. In this case, the
maximum number of quantum predictions that a deterministic local theory can
simultaneously satisfy is $6$: $3$ out of $XZZ=1$, $ZXZ=1$, $ZZX=1$, $XXX=-1$%
, plus the other $3$ ($IYY=1$, $YIY=1$, and $IYY=1$). Equivalently, $\alpha
(H(K_{3}))$ is the maximum quantum violation, denoted by $\beta _{\mathrm{QM}%
}$, minus the minimum number of quantum predictions which \emph{cannot} be
satisfied by a deterministic local theory. Since the minimum number of
quantum predictions which cannot be satisfied by a deterministic local
hidden variable theory is $\beta _{\mathrm{QM}}$ minus the maximum value of
the Bell operator for a deterministic local theory, denoted by $\beta _{%
\mathrm{LHV}}$, and all of them divided by two, then%
\begin{equation}
\alpha (H(K_{n}))=\frac{\beta _{\mathrm{QM}}(n)+\beta _{\mathrm{LHV}}(n)}{2}.
\label{aln}
\end{equation}%
This expression is very useful since, for the Bell inequalities for the $n$%
-qubit GHZ states \cite{gunhe04},
\begin{align}
\beta _{\mathrm{QM}}(n)& =\mu _{\mathrm{QM}}(n)+\tau _{\mathrm{QM}}(n),
\label{des} \\
\beta _{\mathrm{LHV}}(n)& =\mu _{\mathrm{LHV}}(n)+\tau _{\mathrm{LHV}}(n).
\notag
\end{align}%
The interesting point is that the values of $\mu _{\mathrm{QM}}(n)$ and $\mu
_{\mathrm{LHV}}(n)$ are well-known \cite{mermin90}, and
\begin{equation}
\tau _{\mathrm{QM}}(n)=\tau _{\mathrm{LHV}}(n)=\beta _{\mathrm{QM}}(n)-\mu _{%
\mathrm{QM}}(n)
\end{equation}
(Recall that $\beta _{\mathrm{QM}}(n)=2^{n}-1$.) For $n$ odd, $\mu _{\mathrm{%
QM}}(n)=2^{n-1}$, and $\mu _{\mathrm{LHV}}(n)=2^{(n-1)/2}$; for $n$ even, $%
\mu _{\mathrm{QM}}(n)=2^{n-1}$ and $\mu _{\mathrm{LHV}}(n)=2^{n/2}$.
Inserting these numbers in Eq. (\ref{des}) and then in Eq. (\ref{aln}), we
obtain the statement.
\end{proof}

\section{Zero-error capacity} \label{sec:capacity}

In this section, we show that for every graph $G$ on $n$ vertices the graph $%
H(G)$ has zero-error entanglement-assisted capacity $2^{n}-1$. Theorem \ref%
{thm:Hparams} states that $\alpha (H(G))<2^{n}-1$. The result gives a
separation between $c_{0}^{\ast }(H(G))$ and $c_{0}(H(G))$.
It is known that for all graphs, the Lov\'asz number upper bounds the entanglement-assisted Shannon capacity \cite{DSW10}.
Therefore, $c_{0}^{\ast }(H(G))$ saturates its upper bound.

There are few (and very recently discovered) classes of graphs for which this separation is known.
For example, one is based on the Kochen-Specker theorem \cite{CLMW10} and other ones are based on variations of orthogonality graphs
\cite{BBG12, MSS12}.
Here, we present a new family of graphs and a construction method, which can also be interpreted as a graph theoretic technique of independent interest. The most important point is that every graph gives rise to a member of the family through our construction. This property opens directions for future studies, for example, identifying subclasses or hierarchies where the separation is large or is easy to quantify.

\begin{theorem} \label{thm:Capacity}
Let $H$ be a graph from a Kotzig orbit. Then $c_{0}^{\ast }(H)=2^{n}-1$.
\end{theorem}

\begin{proof}
From Theorem \ref{thm:Hparams} and \cite[Corollary 14]{DSW10} we obtain the upper bound
$$  c_{0}^{\ast }(H(G)) \leq \vartheta(H(G)) = 2^{n}-1.
$$
We need to show a matching lower bound on $c_{0}^{\ast }(H(G))$.
We do this by exhibiting a strategy for entangled parties to send one out of $2^{n}-1$
messages in the zero-error setting through a channel with confusability
graph $H$. The strategy is as follows. Alice and Bob share a maximally
entangled state of local dimension $2^{n}(n-1)$. Observe that $H$ can be
partitioned into $2^{n}-1$ cliques, one for each element of the stabilizer
group. The clique corresponding to $s_{i}\in S$ consists of the vertices
associated with the mutually exclusive events in the set $S_{i}$; we denote
by $S_{i}$ the set of events related to $s_{i}$ as in Section \ref{con}. For
each $i\in \{1,2,\ldots,2^{n}-1\}$, Alice performs a projective measurement on
her part of the shared state. The outcomes of the measurement are the
elements of $S_{i}$. Since the parties share a maximally entangled state,
Alice's strategy has to satisfy two properties to be correct:

\begin{enumerate}
\item For each $i\in \{1,2,\ldots,2^{n}-1\}$, the projectors associated to
elements of $S_{i}$ form a projective measurement (because Alice needs to
perform a projective measurement for each message $i$ to be sent).

\item For each edge $\{u,v\}\in E(H(G))$, projectors associated with $u$ and
$v$ must be orthogonal (to satisfy the zero-error constraint).
\end{enumerate}

The next step is to exhibit projectors in Alice's strategy and show that
both properties are satisfied. In what follows we use the notation in
Definition \ref{def:orth}.

We begin by examining the case where $s_{i}$ does not contain any identity
operator. In this case, each projective measurement will consist of
projectors of rank $1$ acting on $\mathbb{C}^{2^n (n-1)}$. Order the elements
of $S_{i}$ arbitrarily. Let $s_{i}$ be of the form $O^{(1)}\cdots O^{(n)}$, where $%
O^{(k)}\in \{X,Y,Z\}$. Define for each $s^{(k)}_{(i,j)}$ the \emph{occurrence
number} $\nu (i,j,k)$ based on a chosen ordering: if the same eigenvector of
$O^{(k)}$ occurs in $s^{(k)}_{(i,j)}$ for the $\ell $-th time in the chosen
ordering then $\nu (i,j,k)=\ell $. Construct projectors
starting from the canonical orthogonal representation and an ancillary space
of dimension $n-1$. For $s^{(k)}_{(i,j)}$, let
\begin{equation}
P_{(i,j)}=\bigotimes_{k=1}^{n}|s^{(k)}_{(i,j)}\rangle \langle
s^{(k)}_{(i,j)}|\otimes |\nu (i,j,k)\rangle \langle \nu (i,j,k)|.
\end{equation}
We show that Property 1 is satisfied.
These projectors are mutually orthogonal for all vertices $(i,j)$. We need
to prove that their sum is the identity. From the structure of the events in
$S_{i}$ we observe that, for each $O_{k}$, the eigenvectors with eigenvalue $%
+1$ (and $-1$) occur in half of the elements of $S_{i}$. Therefore, in the
construction of the projectors, a pair of $\pm 1$ eigenvectors for each $%
O_{k}$ is summed for each ancillary subspace. The sum of each subspace is
the identity. Hence, the total sum is the identity for the whole space.
We show now that also Property 2 is satisfied.
If two projectors are in the same clique, orthogonality follows from the discussion above.
Consider now two projectors of adjacent vertices from two different cliques
that project to the same ancillary subspace. Since we started from an orthogonal
representation, those projectors are orthogonal.

Now, consider the more general case where $s_{i}$ can contain identity
operators. Let $s_{i}$ be of the form $O^{(1)}\cdots O^{(n)}$, where $O^{(k)}\in
\{I,X,Y,Z\}$. We assume that $s_{i}$ has weight $w$. First consider the case where
the first $w$ operators are different from identity, $O^{(1)},O^{(2)},\ldots ,O^{(w)}\neq I$.
To construct the projective measurement for $S_{i}$, we initially construct the
projectors for the first $w$ operators as in the previous case. We obtain
rank-$1$ projectors acting on $\mathbb{C}^{2^w(w-1)}$. Choose a basis for
$\mathbb{C}^{2^n(n-1)-2^w(w-1)}$ and let the projectors be
\begin{equation}
Q_{(i,j)}=\sum_{\ell =1}^{2^n(n-1)-2^w (w-1)}P_{(i,j)}\otimes |\ell \rangle \langle \ell |.
\end{equation}
This ensures that the dimensions match and that Properties $1$ and $2$ hold.
To finish the proof, we need to prove the general case where identity operators
are in arbitrary positions and not all at the end.
In this case, split the construction into subspaces so that each subspace
has all the identities at the end.
Obtain the projectors for the subspaces as described above and then obtain the
final projectors by making tensor products of the projectors for the subspaces.
\end{proof}

We immediately have the following corollary from Theorems \ref{thm:Hparams}
and \ref{thm:Capacity} and the Lov\'asz number upper bound on $\Theta^*$.
\begin{corollary}
Let $H$ be a graph from a Kotzig orbit. Then, $c_0(H)<c_0^*(H)=\Theta^*(H)$.
\end{corollary}

\section{Classification}

Let $\lambda_H$ be the number of connected
components of $H$ - we follow $\lambda_H$ by the number of vertices in each
connected component, \emph{e.g.}, $\lambda_H = 2[6,16]$ means that there are two
connected components, one with 6 vertices and one with 16 vertices.
The degree sequence of $H$, $D_H$, is written as $a,b/c,d/e,f/\ldots$, meaning
that there are $b$ vertices of degree $a$, $d$ vertices of degree $c$, $f$
vertices of degree $e$, etc.

\vspace{2mm}

For each Kotzig orbit, we wish to compute the independence number, $\al(H)$. This is
also listed in Table \ref{BigClass}. Each stabilizing operator of $S$ is
multiplied by
a global coefficient $+1$ or $-1$. For instance, for the operators of $H(P_3)$
(see Fig \ref{H1} and associated discussion), there are six `$+1$' coefficients and one `$-1$'
coefficient.
If one selects one vertex from each of the 6 cliques in $H$ generated by the 6
operators
of this example which have a `$+1$' coefficient, then one can be
sure that they are mutually unconnected in $H$.
(More specifically, one can select for each operator the event where all the measurements
gave outcome $+1$. This event is always present if the operator has coefficient $+1$.)
 So $\al(H) \ge 6$, \emph{i.e.} the
number of `$+1$' coefficients
in $S$ for a given $G$ yields a lower bound on $\al(H)$, for each $G$.
This idea leads to the following lemma.
Let $\beta(G)$ be the number of operators in $S$ with a `$-1$' coefficient where, in general,
$\beta(G)$ is not an invariant of the Kotzig orbit of $G$. Let
$\beta_{\tt{min}}(G^L) = \min \{\beta(G), \hspace{2mm} G \in G^L\}$ and
$\beta_{\tt{max}}(G^L) = \max \{\beta(G), \hspace{2mm} G \in G^L\}$, where $G^L$ is the set of
graphs in the Kotzig orbit of $G$.

\begin{lemma} Given a graph $G$, we have
\begin{equation}
\al(H) \ge 2^n - 1 - \beta_{\tt{min}}(G^L) \ge 2^n - 1 - \beta(G).
\label{betalem}
\end{equation}
\end{lemma}

Table \ref{BigClass} lists both the computed independence number, $\al(H)$, for
each $H$, and the
range of lower bounds, $\beta_{\tt{min}}(G^L) - \beta_{\tt{max}}(G^L)$, on $\al(H)$. One
observes that the lower bound is often tight but not always. For example, for
the graph $G = 15,25,34,45$, computations show that $\beta_{\tt{min}}(G^L) = 6$ and
$\beta_{\tt{max}}(G^L) = 12$. So
$\al(H) \ge 2^5 - 1 - 6 = 25$. In this case the bound is tight as $\al(H)$ is
computed to be $25$.


{
\begin{sidewaystable}
\fontsize{7pt}{0}
\centering
\centering
\caption{$H(G)$ for $n = 2,\ldots,6$}
\label{BigClass}
\begin{equation*}
\begin{array}{|l|l|l|l|l|l|l|l|}
\hline
n & G & |G^L| & |V^H| & \lambda_H & \al(H) & \beta_{\tt{min}}(G^L) - \beta_{\tt{max}}(G^L) & D_H \\ \hline
2 & 12 & 1 & 6 & 3[2,2,2] & 3 & 0 & 1,6/ \\ \hline
3 & 12,13 & 2 & 22 & 2[6,16] & 6 & 1-1 & 3,6/9,16 \\ \hline
4 & 14,24,34 & 2 & 84 & 2[20,64] & 13 & 4-4 & 11,12/13,2/17,6/43,64 \\
& 14,23,24,34 & 4 & 76 & 1 & 13 & 2-4 & 15,4/25,4/27,4/29,32/35,32 \\ \hline
5 & 15,25,35,45 & 2 & 316 & 2[60,256] & 25 & 10-10 & 33,20/45,10/49,30/195,256/
\\
& 15,25,34,45 & 6 & 280 & 1 & 25 & 6-12 &
49,6/51,2/89,12/91,12/107,24/129,72/133,24/159,128/ \\
& 15,24,25,34,35 & 10 & 268 & 1 & 25 & 6-12 &
55,4/99,16/101,4/103,4/107,16/117,64/129,32/137,32/139,32/147,64/ \\
& 12,13,24,35,45 & 3 & 256 & 1 & 25 & 6-10 & 99,40/117,120/135,96 \\ \hline
6 & 16,26,36,46,56 & 2 & 1206 & 2[182,1024] & 51 & 20-20 &
101,30/143,30/147,90/151,2/167,30/841,1024/ \\
& 16,26,36,45,56 & 6 & 1030 & 1 & 51 & 16-28 &
145,12/189,2/207,2/211,6/299,24/301,24/343,4/351,28/ \\
&  &  &  &  & &  & 383,32/455,32/459,96/587,256/665,512/ \\
& 16,26,35,45,46 & 16 & 976 & 1 & 49 & 14-30 &
169,6/207,2/349,12/351,12/357,24/383,8/411,48/449,8/ \\
&  &  &  &  & &  & 453,24/471,32/473,16/477,48/481,48/483,48/533,128/557,128/ \\
&  &  &  &  & &  & 565,64/567,64/611,256/ \\
& 15,26,36,45,56 & 4 & 1044 & 1 & 49 & 14-30 &
173,12/321,36/323,36/479,192/541,64/545,192/679,512/ \\
& 15,26,35,36,45,46 & 5 & 958 & 1 & 51 & 20-28 &
213,6/355,32/395,12/397,12/441,96/515,384/547,192/ \\
&  &  &  &  & &  & 549,192/559,4/567,28/ \\
& 16,24,35,46,56 & 10 & 958 & 1 & 51 & 12-28 &
181,4/213,2/323,4/325,4/347,8/349,8/379,32/401,64/455,8/ \\
&  &  &  &  &  &  & 459,24/477,32/487,64/489,64/515,256/593,256/599,32/603,96/
\\
& 15,26,34,35,46,56 & 25 & 940 & 1 & 49 & 14-30 &
187,4/337,4/339,4/357,16/379,16/383,32/431,64/435,32/ \\
&  &  &  &  &  &  & 469,32/471,32/493,32/495,32/497,128/521,128/529,64/531,64/
\\
&  &  &  &  &  &  & 575,128/601,64/603,64/ \\
& 16,24,26,35,36,45 & 21 & 922 & 1 & 47 & 16-28 &
193,2/355,16/357,16/405,16/413,64/425,64/441,8/475,32/ \\
&  &  &  &  &  &  & 477,32/479,64/503,256/511,32/513,32/557,256/563,8/567,24/ \\
& 14,25,36,45,46,56 & 5 & 990 & 1 & 47 & 20-28 &
197,6/355,12/357,12/433,192/473,96/487,8/491,24/579,192/ \\
&  &  &  &  &  &  & 581,192/625,256/ \\
& 12,13,25,36,45,46 & 16 & 904 & 1 & 45 & 18-30 &
355,32/395,24/407,96/431,48/485,384/539,192/565,64/567,64/ \\
& 12,13,14,23,25,36,45,46,56 & 2 & 936 & 1 & 45 & 22-26 & 427,360/571,576/ \\
\hline
\end{array}%
\end{equation*}
\end{sidewaystable}
}



\vspace{3mm}

The \emph{symmetrical form} of $K_n$ makes it relatively easy to prove that,
for $n > 2$, $H(ST_n) = H(K_n)$ is comprised of two disjoint subgraphs, and the
results of section \ref{Examples} allow us to identify these two disjoint subgraphs, \emph{i.e.}
$H(K_n) = H(\mu_n) + H(\tau_n)$, where $H(\mu_n)$ has $2^{2n-2}$ vertices and
$H(\tau_n)$ has $\sum_{i=1}^{\lfloor \frac{n}{2} \rfloor} 2^{2i-1} \binom{n}{2i}$ vertices -
(see the equations in (\ref{aln}) for the case of $n = 3$).
It is likewise easy to show that $\beta(K_n) = \sum_{i=1}^{\lfloor \frac{n+1}{4} \rfloor}
\binom{n}{4k-1}$, and therefore we know that $\al(H(K_n)) \ge 2^n - 1 - \sum_{i=1}^{\lfloor \frac{n+1}{4} \rfloor}
\binom{n}{4k-1}$.

\vspace{3mm}

The evaluation of $\beta(G)$ can be translated to the following problem.

\begin{lemma}
Define the Boolean function, $f_G(z_1,z_2,\ldots,z_n) : \F_2^n \rightarrow \F_2$ such that
\begin{equation}
f_G(z_1,z_2,\ldots,z_n) = \sum_{\{i,j\},\{j,k\} \in E(G), i < j < k} z_iz_jz_k.
\end{equation}
Let $\text{wt}(f)$ be the \emph{weight} of $f$, defined to be the number of ones in its truth-table).
Then,
\begin{equation}
\beta(G) = \text{wt}(f_G).
\label{cubicwt}
\end{equation}
\label{cubic}
\end{lemma}
\begin{proof}
(Sketch) Consider the subgraph of $G$ at vertices $i,j,k$, where we assume $\{i,j\},\{j,k\} \in E(G)$.
One can confirm that $s = g_ig_jg_k$ is an operator with a global coefficient of $-1$ and
we can represent this in $f_G$ by the cubic term $z_iz_jz_k$.
Moreover this must be true for each such pair of edges in $G$. Likewise consider the five
vertices $i,j,k,l,m$ in $G$ where we assume that $\{i,j\},\{j,k\},\{k,l\},\{l,m\} \in E(G)$.
One can confirm that $s = g_ig_jg_kg_lg_m$ is an operator with a global coefficient of $-1$,
and the operators $s' = g_ig_jg_k$, $s'' = g_jg_kg_l$, and $s''' = g_kg_lg_m$ also have global
coefficients equal to $-1$. We represent this situation in $f_G$ by the sum of the cubic terms
$z_iz_jz_k + z_jz_kz_l + z_kz_lz_m$. The lemma follows from an elaboration of this argument.
Specifically the global coefficient of $s = \prod_{j=1}^n g_j^{z_j}$
is $(-1)^{f_G(z_1,z_2,\ldots,z_n)}$.
\end{proof}

\vspace{2mm}

Lemma \ref{cubic} allows us to obtain an equation for $\beta(C_n)$. We evaluate $\beta(C_n)$
by computer for $n = 3,4,5,6,7,8,\ldots $ to be $1,4,6,18,36,80,\ldots $, respectively and
use \cite{oeis} to find the sequence A051253 and the reference \cite{Cus02} where a
recurrence formula is provided for
$\beta(C_n) = \text{wt}(f_G) = \text{wt}(z_1z_2z_3 + z_2z_3z_4 + \ldots + z_{n-2}z_{n-1}z_n
+ z_{n-1}z_nz_1 + z_nz_1z_2)$, namely
\begin{equation}
\beta(C_{n+3}) = 2(\beta(C_{n+1}) + \beta(C_n) + 2^{n-1}).
\end{equation}

\vspace{3mm}

We also offer the following conjecture, based on the results of Table \ref{BigClass}.
\begin{conjecture}
The graph $H(G)$ is always connected except when $G$ is in the Kotzig orbit of the star graph,
$ST_n$, in which
case $H(G)$ splits into 3 disjoint components for $n = 2$, and $2$ disjoint
components for $n > 2$.
\end{conjecture}

The conjecture is verified, computationally, for $n = 2,3,4,5,$. A potential way to prove it is to try to construct two connected components by adding codewords and show
that this forces one to be in the $ST_n$ orbit.


\vspace{3mm}

By using code-theoretic techniques, we can lower and upper bound the size of cliques in $H(G)$. Remember that $S$ is
the stabilizer
group generated by $G$, comprising $2^n-1$ operators, $s_j$, and that $w_j =
w(s_j)$ is the
{\em weight} of operator $s_j$. Let $w_{\text{max}}(S) = \text{max}_{s_j \in
S}(w(s_j))$ and
$w_{\text{min}}(S) = \text{min}_{s_j \in S}(w(s_j))$. We have that
$w_{\text{max}}(S) = n$, for any $S$, because, for any graph $G$, then
$\prod_{i=0}^{n-1} g_i$
is always an operator of weight $n$. It is also well-known that the stabilizer
group for a graph state
characterizes a self-dual additive code over $\F_4$ of length $n$ whose minimum
distance, $d_H$, is
given by $w_{\text{min}}(S)$ \cite{DP06}. So, in subsequent discussions, we refer to
$w_{\text{min}}(S)$ and
$w_{\text{max}}(S)$ by $d_H$ and $n$, respectively.
Let ${\cal{C}}(H)$ be the set of maximal cliques in $H$, where each $c \in
{\cal{C}}(H)$ is a subset
of $V(H)$, the set of vertices of $H$, over which there is a clique in $H$. Let
$\tilde{\omega}(H)$ and
$\omega(H)$ be the minimum and maximum size of a maximal clique, respectively,
in graph $H$. So
$\tilde{\omega}(H) = \text{min}\{|c| \hspace{2mm} : \hspace{2mm} c \in
{\cal{C}}(H)\}$ and
$\omega(H) = \text{max}\{|c| \hspace{2mm} : \hspace{2mm} c \in {\cal{C}}(H)\}$.
Recall that $\omega(H)$ denotes the clique number of $H$.

\begin{theorem}\label{cliqbound}
For $H(G)$, the graph generated from the stabilizer set $S$,
\begin{equation}
\begin{array}{ll}
 \tilde{\omega}(H) \le |c| \le \omega(H), & \hspace{10mm} \forall c \in {\cal{C}}(H), \\
 & \text{where } \tilde{\omega}(H) = 2^{d_H - 1}, \omega(H) = 2^{n - 1}, \\
 & \text{and both upper and lower bounds are tight.}
\end{array}
\end{equation}
\end{theorem}
\begin{proof}
Consider an arbitrary set of two operators, $R = \{I XZYZ, -YXZYY\}$.
Then both operators have $XZY$ at their second,
third, and fourth tensor positions. We say that the two operators have a set
overlap of $XZY$ and this overlap is of size $\mu_R = 3$.
Let $R \subset S$ and consider an arbitrary splitting of the
assignments to $XZY$ of $xzy$, $\underline{x}z\underline{y}$, $\underline{x}zy$,
$x\underline{z}\underline{y}$, $\underline{x}\underline{z}\underline{y}$ for
operator
$I XZYZ$ and $x\underline{z}y$, $\underline{x}\underline{z}y$, $xz\underline{y}$
for operator $-YXZYY$. Then $H(G)$ contains a size-11 clique over
the vertices $Ixzyz,I\underline{x}zy\underline{z},I\underline{x}z\underline{y}z,
Ix\underline{z}\underline{y}z,I\underline{x}\underline{z}\underline{y}\underline{z}$
and
$yx\underline{z}yy,\underline{y}\underline{x}\underline{z}yy,yxz\underline{y}y,
\underline{y}x\underline{z}y\underline{y},y\underline{x}\underline{z}y\underline
{y},\underline{y}xz\underline{y}\underline{y}$.
It is straightforward to verify that
this clique cannot
be extended by adding another vertex from the subset of vertices in $H$
originating from the operators in $R$.
More generally, any splitting of the $2^3$ assignments to $XZY$ is
possible, each giving
a different clique. For our example, the operator $I XZYZ$ contributes
$5 \cdot 2^{w(I XZYZ) - \mu_R - 1}$ vertices to the clique and operator $-YXZYY$
contributes
$3 \cdot 2^{w(-YXZYY) - \mu_R - 1}$ vertices to the clique. So this clique is of size
$2^{-\mu_R - 1}(5.2^{w(I XZYZ)} + 3 \cdot 2^{w(-YXZYY)}) = 11$. It is a maximal clique if and only if
there is no $ R' \subset S$ such that $R \subset R'$ and $\mu_{R'} = \mu_R$.

If, instead, $R = \{I XZYZ, -YXZYY, XXZXI \}$, then the set
overlap is reduced to $XZ$ and of size $\mu_R = 2$. As an example, consider
just the partition
where $\underline{x}z$ is assigned to operator $I XZYZ$,
$xz,\underline{x}\underline{z}$ to
operator $-YXZYY$, and $x\underline{z}$ to operator $XXZXI $. Then this
assignment identifies
the vertices $I\underline{x}z\underline{y}z,I\underline{x}zy\underline{z}$, and
$\underline{y}xzyy,yxz\underline{y}y,yxzy\underline{y},\underline{y}xz\underline{y}\underline{y},
\underline{y}\underline{x}\underline{z}yy,y\underline{x}\underline{z}\underline{y}y,y\underline{x}\underline{z}y\underline{y},
\underline{y}\underline{x}\underline{z}\underline{y}\underline{y}$, and
$\underline{x}x\underline{z}xI,xx\underline{z}\underline{x}I$.
The operator $I XZYZ$ contributes
$1 \cdot 2^{w(I XZYZ) - \mu_R - 1}$ vertices to the clique, operator $-YXZYY$
contributes
$2 \cdot 2^{w(-YXZYY) - \mu_R - 1}$ vertices to the clique, and operator $XXZXI$
contributes
$1 \cdot 2^{w(XXZXI ) - \mu_R - 1}$ vertices to the clique. So this clique is of size
$2^{-\mu_R - 1}(1 \cdot 2^{w(I XZYZ)} + 2 \cdot 2^{w(-YXZYY)} + 1 \cdot 2^{w(XXZXI )}) = 12$. It
is a maximal clique if and only if
there is no $ R' \subset S$ such that $R \subset R'$ and $\mu_{R'} = \mu_R$.

Therefore, a strategy to find all maximal cliques in $H(G)$ is to find all subsets $R \subset S$
where $\mu_R > 0$, and such that there is no $ R' \subset S$ where $R \subset R'$ and
$\mu_{R'} = \mu_R$. Then for each of these subsets, $R$, and for each $|R|$-wise
ordered partition,
$p = \{a_{s_j} \hspace{2mm} | \hspace{2mm} s_j \in R\}$ of $2^{\mu_R}$,
where $\sum_{s_j \in R} a_{s_j} = 2^{\mu_R}$, and for each assignment of the
integers in set $\{0,1,\ldots,2^{\mu_R}-1\}$ according
to partition $p$, there is associated a maximal clique, $K_{R,p}$, of size,
\begin{equation}
|K_{R,p}| = 2^{-\mu_R - 1}\sum_{s_j \in R} a_{s_j}2^{w(s_j)}.
\label{maximalsize}
\end{equation}

Let $w^-(s_j) = \text{min}\{w(s_j) \hspace{2mm} | \hspace{2mm} s_j \in R \}$.
Then, for the specific $s_j \in R$ where $w(s_j) = w^-(s_j)$ we can minimise the
clique size by assigning
$a_j = 2^{\mu_R}$. So a special case of the above equation identifies a maximal
clique of minimum
size $2^{w^-(s_j)-1}$. Similarly
let $w^+(s_j) = \text{max}\{w(s_j) \hspace{2mm} | \hspace{2mm} s_j \in R \}$.
Then, for the specific $s_j \in R$ where $w(s_j) = w^+(s_j)$ we can maximise the
clique size by assigning
$a_j = 2^{\mu_R}$. A special case of the above equation identifies a maximal
clique of maximum
size $2^{w^+(s_j)-1}$. We arrive at both $\tilde{\omega}(H) = 2^{d_H - 1}$ and
$\omega(H) = 2^{n - 1}$ by
observing that all members of $S$ occur in at least one $R$ identified.
\end{proof}

\vspace{2mm}

In a step towards enumerating the maximal cliques of $H(G)$, we first enumerate
the number of maximal cliques of size given by (\ref{maximalsize}), as
generated by a fixed $R$ and $|R|$-wise ordered partition, $p$.
Let $R_k \subset R$ satisfy $R_k = \{s_j \in R \hspace{2mm} | \hspace{2mm} j <
k\}$, and let $b_k = \sum_{s_j \in R_k} a_j$.

\begin{lemma}
For fixed $R$ and $p$, the graph $H(G)$, contains $\# K_{R,p}$ maximal cliques
of size $|K_{R,p}|$, where
\begin{equation}
\# K_{R,p} = \prod_{s_j \in R} \binom{2^{\mu_R - b_j}}{a_j}.
\end{equation}
\end{lemma}
\begin{proof}
The lemma follows immediately by counting the number of ways that one can assign to the
integers an $|R|$-way partition of $2^{\mu_R}$. For instance, for $\mu_R = 2$
and partition
$\{1,2,1\}$, one can assign the integers as $\{\{0\},\{1,2\},\{3\}\}$,
$\{\{0\},\{1,3\},\{2\}\}$, $\{\{0\},\{2,3\},\{1\}\}$, $\{\{1\},\{0,2\},\{3\}\}$,
$\{\{1\},\{0,3\},\{2\}\}$, $\{\{1\},\{2,3\},\{0\}\}$, $\{\{2\},\{0,1\},\{3\}\}$,
$\{\{2\},\{0,3\},\{1\}\}$, $\{\{2\},\{1,3\},\{0\}\}$, \\
$\{\{3\},\{0,1\},\{2\}\}$,
$\{\{3\},\{0,2\},\{1\}\}$, $\{\{3\},\{1,2\},\{0\}\}$ - in total
$\binom{2^2}{1} \times \binom{2^2-1}{2} \times \binom{2^2-3}{1} = 4 \times 3
\times 1 = 12$ ways.
\end{proof}

\vspace{5mm}

\vspace{3mm}

Although, in Theorem \ref{cliqbound},
we have obtained tight lower and upper bounds for the size of the maximal cliques of $H(G)$ in the
general case,
in terms of $d_H$, and $n$, respectively, it remains open to obtain tight lower and upper bounds for the
size of the maximal independent sets of $H(G)$ in the general case.
However we have obtained a lower bound, $\beta(G)$,
on $\al(H(G))$, although this bound is not tight in general. It also remains open to
provide equations for $|V^H|$ in the general case. Given that this paper highlights
the gap between $\al(H(G))$ and $2^n - 1$, it is particularly desirable to develop equations for
$\al(H(G))$ in the general case. Of course, it would be nice to find a formula linking $\al(H(G))$ and $\al(G)$ - however it should be noted that $\al(G)$ is not an invariant of the Kotzig orbit of $G$.

\begin{paragraph} {Acknowledgements} We would like to thank Mary Beth Ruskai for asking during the Quantum Information Workshop at the Centro de Ciencias de Benasque (July 2011) the question that motivated the results: Is there a connection between graph states and the graphs representing exclusivity structures studied in \cite{cabello10}? We thank the Associate Editor and a referee for the very useful comments that improved the readability of the paper. AC is supported by the Project No.\ FIS2011-29400 (Spain). Part of this work has been done while AC and SS where visiting the Department of Informatics at the University of Bergen. The financial support of the University of Bergen is gratefully acknowledged. Part of this work has been done while GS was visiting University College London. GS is supported by Ronald de Wolf's Vidi grant 639.072.803 from the Netherlands Organization for Scientific Research (NWO).

\end{paragraph}

\end{document}